%% file: main.tex
\documentclass[submission,copyright,creativecommons]{eptcs}
\providecommand{\event}{GandALF 2022} 
\usepackage{breakurl}             
\usepackage{underscore}           

\usepackage[utf8]{inputenc}
\usepackage[T1]{fontenc}
\usepackage{microtype}
\usepackage{booktabs}
\usepackage{tikz}
\usepackage{xspace}
\usepackage[algoruled]{algorithm2e}
\usetikzlibrary{shapes,arrows,positioning}

\usepackage[textwidth=2.5cm]{todonotes}
\usepackage[shortlabels]{enumitem}

\usepackage{amstext}
\usepackage{amsmath}
\usepackage{amsfonts}

\usepackage{amssymb}

\newtheorem{lemma}{Lemma}
\newtheorem{theorem}{Theorem}
\newtheorem{definition}{Definition}
\newtheorem{proposition}{Proposition}
\newenvironment{proof}[1][Proof]{\begin{trivlist}
\item[\hskip \labelsep {\bfseries #1}]}{\end{trivlist}}

\usepackage{comment}

\usepackage{colortbl}
\usepackage{slashbox}

\usepackage{times}

\input{macros}

\title{Analyzing Robustness of Angluin's L$\!^*$ Algorithm\\ ~in Presence of Noise}

\author{Igor Khmelnitsky  
\institute{Université Paris-Saclay \\CNRS, ENS Paris-Saclay\\INRIA, LMF, France}
\email{igor.khme@gmail.com }
\and 
Serge Haddad
\institute{Université Paris-Saclay \\CNRS, ENS Paris-Saclay\\INRIA, LMF, France}
\email{haddad@lsv.fr}
\and 
Lina Ye
\institute{Université Paris-Saclay \\CNRS, ENS Paris-Saclay \\CentraleSupélec, LMF, France}
\email{lina.ye@centralesupelec.fr}
\and 
Benoît Barbot
\institute{Université Paris-Est Créteil\\France}
\email{benoit.barbot@u-pec.fr}
\and 
Benedikt Bollig
\institute{Université Paris-Saclay \\CNRS, ENS Paris-Saclay, LMF, France}
\email{bollig@lsv.fr}
\and
Martin Leucker
\institute{Institute for Software Engineering and \\ Programming Languages \\ Universität zu
Lübeck, Germany}
\email{leucker@isp.uni-luebeck.de}
\and
Daniel Neider
\institute{
Carl von Ossietzky\\ University of Oldenburg\\Germany}
\email{daniel.neider@uol.de}
\and
Rajarshi Roy
\institute{Max Planck Institute\\ for Software Systems \\ Kaiserslautern, Germany}
\email{rajarshi@mpi-sws.org}
}

\def\titlerunning{Analyzing Robustness of Angluin’s L* Algorithm in Presence of Noise}
\def\authorrunning{I. Khmelnitsky et al. }

\begin{document}
\maketitle

\begin{abstract}
          Angluin's L$^\ast$ algorithm learns the minimal (complete)
          deterministic finite automaton (DFA) of a regular language
          using membership and equivalence queries. Its probabilistic
          approximatively correct (PAC) version substitutes an
          equivalence query by a large enough set of random membership
          queries to get a high level confidence to the answer. Thus
          it can be applied to any kind of (also non-regular) device
          and may be viewed as an algorithm for synthesizing an
          automaton abstracting the behavior of the device based on
          observations.  Here we are interested on how Angluin's PAC
          learning algorithm behaves for devices which are obtained
          from a DFA by introducing some noise. More precisely we
          study whether Angluin's algorithm reduces the noise and
          produces a DFA closer to the original one than the noisy
          device.  We propose several ways to introduce the noise: (1)
          the noisy device inverts the classification of words w.r.t.\
          the DFA with a small probability,  (2) the noisy
          device modifies with a small probability the letters of the
          word before asking its classification w.r.t.\ the DFA, and (3) the
          noisy device combines the classification of a word w.r.t.\ the DFA
          and its classification w.r.t. a counter automaton.  Our experiments
          were performed on several hundred DFAs.
          Our main contributions, bluntly stated, consist in showing that: (1)
          Angluin's algorithm behaves well whenever the noisy device 
         is produced by a random process, (2) but poorly with a structured noise, and, that (3)
         almost surely randomness
         yields systems with non-recursively enumerable languages. 
\end{abstract}


\input{intro}
	
\let\ltxtab\tabular
\let\ltxendtab\endtabular
\renewenvironment{tabular}{\setlength{\extrarowheight}{5pt}\setlength{\belowrulesep}{2pt}%
\ltxtab}{\ltxendtab}


\input{preliminaries}


\input{contribution}


\input{evaluation}

	\input{theory}

	\input{conclusion}

\bibliographystyle{eptcs}
\bibliography{bib}

\end{document}

%% file: macros.tex
\newcommand{\nat}{\mathbb{N}}
\newcommand{\integer}{\mathbb{Z}}

\newcommand{\Lan}{\mathcal{L}}	

\newcommand{\cA}{\mathcal{A}}
\newcommand{\cM}{\mathcal{M}}
\newcommand{\cN}{\mathcal{N}}

\definecolor{verylow}{HTML}{e75874}
\definecolor{low}{HTML}{e75874}
\definecolor{medium}{HTML}{ff8c00}
\definecolor{high}{HTML}{53A567}
\definecolor{veryhigh}{HTML}{53A567}

\newcolumntype{"}{@{\hskip\tabcolsep\vrule width 0.8pt\hskip\tabcolsep}}
\newcommand{\qed}{\hfill$\blacksquare$}

\newcommand{\RNN}{R}

\newcommand{\N}{\mathbb{N}}
\newcommand{\A}{\mathcal{A}}

\newcommand{\emptyword}{\lambda}
\renewcommand{\epsilon}{\varepsilon}

\newcommand{\dist}[2]{d(#1,#2)}

%% file: intro.tex
\section{Introduction}
\label{sec:intro}

\paragraph{Discrete-event systems and their languages.}

Discrete-event systems \cite{Cassandras10} form a large class of dynamic systems that, given some internal state, evolve from one state to another one due
to the occurrence of an event. For instance, discrete-event systems can represent a  cyber-physical process whose events are triggered 
by a controller or the environment, or, a business process whose events are triggered by human activities or software executions.
Often, the behaviors of such systems are classified as safe (aka correct, representative, etc.) or unsafe. Since a behavior may be identified
by its sequence of occurred events, this leads to the notion of a language.

\paragraph{Analysis versus synthesis.}

There are numerous formalisms to specify (languages of) discrete-event systems.
From a designer's perpective, the simpler it is the better its analysis will be. So finite automata and their languages (regular languages)
are good candidates for the specification. However, even when the system is specified by an automaton, its implementation may slightly differ
due to several reasons (bugs, unplanned human activities, unpredictable environment, etc.). Thus, one generally checks whether 
the implementation conforms to the specification. However, in many contexts, the system has already been implemented and the original specification
(if any) is lost, as for instance in the framework of process mining \cite{vanderalst12}. Thus, by observing and interacting with the system, one aims to recover a specification close to the system at hand but that is robust with respect to its pathologic behaviors. 

\paragraph{Language learning.}

The problem of learning a language from finite samples of strings by discovering the corresponding grammar is known as grammatical inference. Its significance was initially stated in~\cite{Solomonoff64} and an overview of very first results can be found in~\cite{Biermann72}. As it may not always be possible to infer a grammar that exactly identifies a language, approximate language learning was introduced in~\cite{Wharton1974}, where a grammar is selected from a solution space whose language approximates the target language with a specified degree of accuracy.  To provide a deeper insight into language learning, the problem of identifying a (minimal) deterministic finite automaton   (DFA) that is consistent with a given sample has attracted substantial attention in the literature since several decades~\cite{Gold1978, Angluin87, Valiant84}.  An understanding of regular language learning is very valuable for a generalization to other more complex classes of languages.

\paragraph{Angluin's L$\!^{*}$ algorithm.}

Angluin's L$\!^{*}$ algorithm learns the minimal DFA of a regular language using membership and equivalence queries. Thus, one could try to adapt it to the synthesis task described above. However, for most black box systems, it is almost impossible to implement the equivalence query.  Thus, its probabilistic approximatively correct (PAC) version substitutes an equivalence query with a large enough set of random membership queries. However, one needs to define and evaluate the accuracy of such an approach. Thus, here we are interested in how PAC Angluin's algorithm behaves for devices which are obtained from a DFA by introducing some noise.

\paragraph{Noisy learning.}

Most learning algorithms in the literature assume the correctness of the training data, including the example data such as attributes as well as classification results. However, sometimes noise-free datasets are not available. \cite{Quinlan86} carried out an experimental study of the noise effects on the learning performance. The results showed that generally the classification noise had more negative impact than the attribute one, i.e., errors in the values of attributes.     
\cite{AngluinL87} studied how to compensate for randomly introduced noise and discovered a theorem giving a bound on the  smaple size that is sufficient  for PAC-identification in the presence of classification noise when the concept classes are finite.  
Michael Kearns formalized another related learning model from statistical queries by extending Valiant's learning model~\cite{Kearns98}. One main result shows that any class of functions learnable from this statistical query model is also learnable with classification noise in Valiant's model. 

\paragraph{Our contribution.}

In this paper, we study against which kinds of noise Angluin's algorithm\footnote{In this work by ``Angluin's algorithm'' we refer to the optimized version from \cite{Kearns94}.} is \emph{robust}. To the best of our knowledge, this is the very first attempt of noise analysis in the automata learning setting.
More precisely, we consider the following setting (cf.\ Figure~\ref{fig:General}): Assume that a regular device $\cA$ is given, typically as a black box. Due to some noise $\cN$, the system $\cA$ is pertubed resulting in a not necessarily regular system $\cM_\cN$. This one is consulted by the PAC version of $L^\ast$ to obtain a regular system $\cA_E$. The question studied in this paper is whether $\cA_E$ is closer to $\cA$ than $\cM_\cN$, or, in other words, to which extent learning via $L^\ast$ is robust against the noise $\cN$.
\vspace{1cm}
\begin{figure}[h]
	\centering
	\includegraphics[scale=0.9]{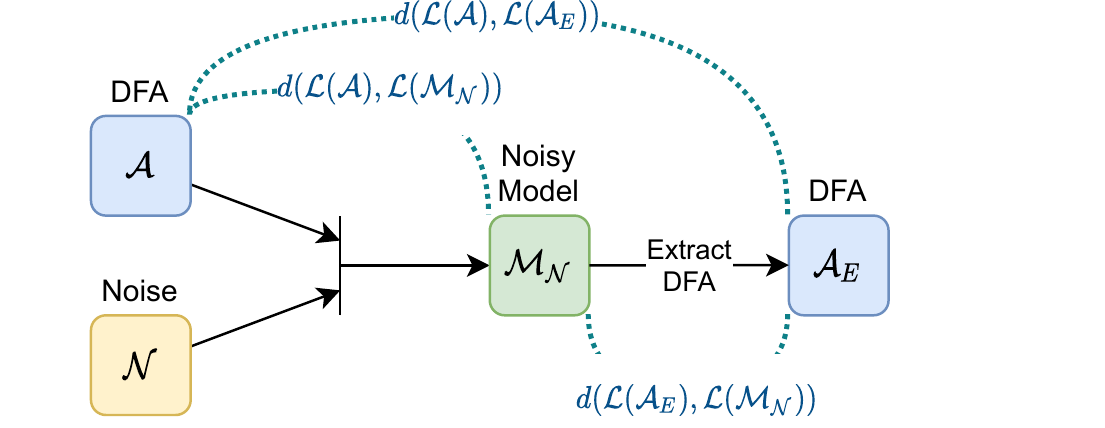}
	\label{fig:General}
	\caption{The experimental setup and the studied distances}
\end{figure}
\vspace{1cm}
To this end, we introduce three kinds of \emph{noisy devices} obtained from the DFA $\cA$: (1) the noisy device is obtained by a random process from a given DFA by 
inverting the classification of words with a small probability, which corresponds to the classification noise in the classical learning setting,  
(2) the noisy device is obtained by a random process that, with a small probability, replaces each letter of a word by one chosen uniformly from the alphabet and then determines its classification based on the DFA, which corresponds 
to the attribute noise in the classical setting,  and
(3) the noisy  DFA combines the classification of a word w.r.t. the DFA and its status w.r.t.\ a counter automaton.
Our studies are based on the distribution over words that is used for generating words associated with membership queries 
and defining (and statistically measuring) the \emph{distance}
between two devices as the probability that they differ on word acceptance. 
We have performed experiments
over several hundreds random DFA. 
We have pursued several goals along our experiments, expressed by the following questions:
\vspace{0.3cm}
\begin{itemize}[]
\setlength\itemsep{0.5em}
	\item What is the threshold (in terms of distance) between pertubating the DFA or producing a device that is no more ``similar to'' the DFA?
	\item What is the impact of the nature of noise on the robustness of  Angluin's algorithm?
	\item What is the impact of the words distribution on the robustness of  Angluin's algorithm?
	
\end{itemize}
\vspace{0.3cm}
Due to the approximating nature of the PAC version of $L^\ast$, we had to consider the question of how to choose the accuracy  of the approximate equivalence query to get a good trade-off between accuracy and efficiency. Moreover, since in most cases,  Angluin's algorithm may perform a huge number of refinement rounds before a possible termination, we considered what a ``good'' number of rounds to stop the algorithm avoiding underfitting and overfitting is.

We experimentally show that w.r.t. the random noise, i.e., the noise introduced with a small probability in different ways, Angluin's algorithm behaves quite well, i.e., the learned DFA ($\cA_E$) is very often closer to the original one ($\cA$) than the noisy random device ($\cM_\cN$). When the noise is obtained using the counter automaton, Angluin's algorithm is not robust. Instead, the device $\cA_E$ is closer to the noisy device $\cM_\cN$. Moreover, we establish that the expectation of the length of a random word should be large enough to cover a relevant part of the set of words in order for Angluin's algorithms to be robust. 

In order to understand why  Angluin's algorithm is robust w.r.t.\ random noise we have undertaken a theoretical study establishing that
almost surely the language of the noisy device  ($\cM_\cN$) for case (1) and, with a further weak assumption, also for case (2) is not recursively enumerable. Considering non-recursively enumerable languages as unstructured, this means that due to the noise, the (regular) structure of $\cA$ vanishes. This is not the case for the counter automaton setting. Altogether, to put it bluntly: the less structure the noisy device has, the better Angluin's algorithm works.

\paragraph{Organization.}

In Section~\ref{sec:preliminaries}, we introduce the technical background required for the robustness analysis.
In Section~\ref{sec:algorithm}, we detail the goals and the settings of our analysis. In Section~\ref{sec:evaluation}, we provide and discuss the experimental results. 
In Section~\ref{sec:theory}, we discuss randomness versus structure. Finally in Section~\ref{sec:con}, we draw our the conclusions and identify future work.

%% file: preliminaries.tex
\section{Preliminaries} 
\label{sec:preliminaries}
 
Here we provide the technical background required for the robustness analysis.

\paragraph{Languages.}

Let $\Sigma$ be an alphabet, i.e., a nonempty finite set,
whose elements are called \emph{letters}. A \emph{word} $w$ over $\Sigma$
is a finite sequence over $\Sigma$, whose length is denoted by $|w|$.
The unique word of length $0$ is called the \emph{empty word} and denoted by $\emptyword$.
As usual, $\Sigma^\ast$ is the set of all words over $\Sigma$,
and $\Sigma^+ = \Sigma^\ast \setminus \{\emptyword\}$ is
the set of words of positive length.
A \emph{language} (over $\Sigma$)
is any set $L \subseteq \Sigma^\ast$.
The symmetric difference of languages $L_1,L_2 \subseteq \Sigma^\ast$
is defined as $L_1 \Delta L_2 = (L_1 \setminus L_2) \cup (L_2 \setminus L_1)$.

\paragraph{Words distribution and measure of a language.}

A distribution $D$ over $\Sigma^*$ is defined by a mapping ${\bf Pr}_D$
from $\Sigma^*$ to $[0,1]$ such that $\sum_{w\in \Sigma^*}{\bf Pr}_D(w)=1$.
Let $L$ be a language. Its probabilistic measure w.r.t. $D$, ${\bf Pr}_D(L)$ 
is defined by ${\bf Pr}_D(L)=\sum_{w\in L}{\bf Pr}_D(w)$.

Our analysis requires that we are able to efficiently  sample a word according to 
some $D$. Thus we only consider  distributions $D_\mu$ with  $\mu\in\ ]0,1[$,
that are defined for a word $w = a_1 \ldots a_n \in \Sigma^\ast$ by
\[{\bf Pr}_{D_\mu}(w) = \mu \left(\frac{1-\mu}{|\Sigma|}\right)^n\,.\]
To sample a random word according to $D_\mu$ in practice, we start with the empty word
and iteratively we flip a biased coin with probability $1-\mu$ to add a letter (and $\mu$ to return the current word)
and then uniformly select the letter in $\Sigma$.

\paragraph{Language distance.}

Given languages $L_1$ and $L_2$, their distance w.r.t.\ a distribution $D$,
$d_D(L_1,L_2)$, is defined by $d_D(L_1,L_2)={\bf Pr}_D(L_1 \Delta L_2)$.  Computing the distance between
languages is in most of the cases impossible. Fortunately whenever the membership problem
for $L_1$ and $L_2$ is decidable, then using Chernoff-Hoeffding bounds~\cite{Hoff63}, this distance can be statistically approximated as follows.
Let $\alpha,\gamma > 0$ be an error parameter and a confidence level, respectively.
Let $S$ be a set of words sampled independently according to $D$, called a sampling, such that  $|S| \geq \frac{\log(2 / \gamma )}{ 2\alpha^2}$.
Let
$dist =\frac{|S \cap (L_1 \Delta L_2)|}{|S|}$.
Then, we have
\[
{\bf Pr}_D(|d_D(L_1,L_2) - dist| > \alpha ) ~<~ \gamma\,.
\]
Since we will not simultaneously discuss about multiple distributions,
we omit the subscript $D$
almost everywhere.

\paragraph{Finite Automata.}

A (complete) deterministic finite automaton (DFA) over $\Sigma$ is a tuple
$\A = (Q,\sigma,q_0,F)$ where $Q$ is a finite set
of states, $q_0 \in Q$ is the initial state,
$F \subseteq Q$ is the set of final states, and
$\sigma: Q \times \Sigma \to Q$ is the transition function.
The transition function is inductively extended over words by 
 $\sigma(q,\emptyword)=q$ and  $\sigma(q,wa)=\sigma(\sigma(q,w),a)$.
The language of $\A$ is defined as
$\Lan(\A) = \{w \in \Sigma^\ast \mid \sigma(q_0,w) \in F\}$.
A language $L \subseteq \Sigma^\ast$ is called \emph{regular} if
$L = \Lan(\A)$ for some DFA $\A$.

\newcommand{\tprob}{p}
\newcommand{\ntprob}{p}
\newcommand{\letterprop}{p}
\newcommand{\width}{\epsilon}
\newcommand{\confidence}{\gamma}
\newcommand{\letterprobp}[1]{p_{#1}}

\paragraph{PAC version of Angluin's L$\!^{*}$ algorithm.}

Given a regular language $L$, Angluin's L$\!^{*}$ algorithm learns the unique mimimal DFA $\A$ such that $\Lan(\A)=L$
using only membership queries `Does $w$ belong to $L$?' and equivalence queries `Does $\Lan(\A_E)=L$? and 
if not provide a word $w\in L\Delta \Lan(\A_E)$'. An abstract version of this algorithm is depicted by Algorithm~\ref{algo:PAC}.
The main features of this algorithm are: a data structure $Data$ from which {\tt Synthetize}($Data$) returns an automaton $\A_E$
and such that given a word $w \in L\Delta \Lan(\A_E)$, {\tt Update}($Data,w)$ updates $Data$. The number of states of $\A_E$ is incremented by one
after each round and so the algorithm terminates after its number of states is equal to the (unknown) number of states of $\A$.

The Probably Approximately Correct (PAC) version of Angluin's L$\!^{*}$ algorithm takes as input an error parameter $\varepsilon$ and a confidence level $\delta$, and 
replaces the equivalence query by a number of membership queries `$w \in  L\Delta \Lan(\A)$?' where the words are sampled from some distribution $D$ unknown to the algorithm. Thus this algorithm can stop too early when all answers are negative while $L\neq \Lan(\A)$. However due to the number of such queries which depends on the current round $r$ (i.e., $\lceil \frac{log(1/\delta)+(r+1)\log(2)}{\varepsilon}\rceil$) this algorithm ensures that
\[
{\bf Pr}_D(d_D(L,\Lan(\A)) > \epsilon ) ~<~ \delta\,.
\]

A key observation is that this algorithm could be used for every language $L$ for which the membership problem is decidable. However since
$L$ is not necessarily a regular language, the algorithm might never stop and thus our adaptation includes a parameter $maxround$ that ensures termination.

\begin{algorithm}[h]
        \DontPrintSemicolon
        
        \KwIn{$L$, a language unknown to the algorithm}
        \KwIn{an integer $maxround$ ensuring termination}
        \SetKwFunction{Angluin}{Angluin}

	\BlankLine
	{\Angluin}$()$
    
	\BlankLine
	\KwData{an integer $r$, a boolean $b$ and a data structure $Data$}
	\KwOut{a DFA $\A_E$}
	\BlankLine
	
	{\tt Initialize}($Data$)\;
	$r \leftarrow 0$\;
	\BlankLine
	
	\tcp{The control of $maxround$ is unnecessary when $L$ is regular}
	\While{$r < maxround$}
	{
	 $\A_E \leftarrow $ {\tt Synthetize}($Data$)\;
	 
	 $(b,w) \leftarrow$ {\tt IsEquivalent}($\A_E$)
	  
	\lIf{$b$}{\Return $\A_E$} 
	 {\tt Update}($Data,w$)\;
	 $r\leftarrow r+1$
	}
	
	 \Return  {\tt Synthesize}($Data$)

	\caption{Angluin's L$\!^*$ algorithm}
	\label{algo:PAC}    
\end{algorithm}

%% file: contribution.tex
\section{Robustness Analysis}
\label{sec:algorithm}
\subsection{Principle and goals of the analysis}

\paragraph{Principle of the analysis.}
Figure~\ref{fig:General} illustrates the whole process of our analysis. 
First we set the qualitative and quantitative nature of the noise ($\mathcal N$). 
Then we generate a set of random DFA ($\A$).
Combining $\A$ and $\mathcal N$, one gets a noisy model $\mathcal M_{\mathcal N}$.
More precisely, depending on whether the noise is random or not,  
$\mathcal M_{\mathcal N}$ is either generated off-line (deterministic noise) or on-line
(random noise) when a membership query is asked during Angluin's L$\!^*$ algorithm. 
Finally we compare (1) the distances between $\A$ and $\mathcal M_{\mathcal N}$, 
and (2) between $\A$ and $\mathcal  \A_E$, the automaton returned by
the algorithm. The aim of this comparison is to establish whether
$\A_E$ is closer to $\A$ than $\mathcal M_{\mathcal N}$. 
In order to get a quantitative measure, we define the  \emph{information gain} as: 
\[
	\text{Information gain} = \frac{d(\Lan(\A),\Lan(\mathcal M_{\mathcal N}))}{d(\Lan(\A),\Lan(\A_E))}
\] 
We consider a {\bf \color{low} low} information gain to be in $[0
,0.9)$, a {\bf\color{medium} medium} information gain to be in $[0.9
,1.5)$, and a  {\bf\color{high} high} information gain to be in $[1.5 ,\infty)$. To make high information gain more evident, we set its threshold value as 1.5. 

In addition,
we also evaluate the distance between  $\A_E$ and $\mathcal M_{\mathcal N}$
in order to study in which cases the algorithm learns in fact the noisy device instead of the original DFA.

\paragraph{Goals of the analysis.}

\begin{itemize}[]
	\item {\bf Quantitative analysis.} The information gain highly depends on the `quantity' of the noise, i.e., error rate. So we analyze the information gain 
	depending of the distance between the original DFA and the noisy device and want to identify
	a threshold (if any) where  the information gain starts to significantly increase.
	\item {\bf Qualitative analysis.} Another important criterion of the information gain is the `nature' of the noise. So we analyze the information gain w.r.t. the three noisy devices that we have introduced.
	\item{\bf Impact of word distribution.} Finally, the
          robustness of the L$^\ast$ algorithm with respect to word distribution is also analyzed.
\end{itemize}

 In order to perform relevant experiments, one needs to tune two critical parameters of Angluin's L$^\ast$ algorithm. Since  the running time of the algorithm quadratically depends on
	the number of rounds (i.e. iterations of the loop), selecting an appropriate \emph{maximal number of rounds} is a critical issue.
	We vary this maximal number of rounds and analyze how the information gain decreases
	w.r.t. this number.
As an equivalence query is replaced with a set of membership queries whose number depends
	on the current round and the pair $(\varepsilon,\delta)$, it is thus interesting to study (1) what is the
	effect of \emph{accuracy of the approximate equivalence queries}, i.e., the values of  $(\varepsilon,\delta)$ on the ratio of executions that reach the maximal number of rounds
	and (2) compare the  information gain  for executions that stop before reaching this maximal number
	and the same execution when letting it run up  to this maximal number.

\subsection{Noise}

\newcommand{\RL}{R}

A \emph{random language} $R \subseteq \Sigma^\ast$ is determined by
a random process: for each $w \in \Sigma^\ast$, membership $w \in \RL$
is determined independently at random, \emph{once and for all}, according to
some probability ${\bf Pr}(w \in R) \in [0, 1]$. The probability ${\bf Pr}(w \in R)$
may depend on some parameters such as $w$ itself and a given DFA.

We now describe the three kinds of noise that we analyze in this paper.
Each type adds noise to a given DFA $\A$ in form of a random language $\RL$.
For the first two types, \emph{noise with output} and \emph{noise with input},
the probability ${\bf Pr}(w \in R)$ of including $w \in \Sigma^\ast$ in $\RL$ depends on $w$
itself, $\Lan(\A)$, and some parameter $0<p<1$.
The third kind of noise, \emph{counter DFA}, is actually \emph{deterministic},
i.e., ${\bf Pr}(w \in R) \in \{0,1\}$ for all $w \in \Sigma^\ast$. In that case,
the given DFA $\A$ determines a unique ``noisy'' language.
Let us be more precise:

\paragraph{DFA with noisy output.}

Given a DFA $\A$ over the alphabet $\Sigma$ and $0<p<1$, the random
language $\Lan(\A^{\rightarrow p})$ flips the classification of words w.r.t.\ $\Lan(\A)$ with probability $p$.
More formally, for all $w \in \Sigma^\ast$, $${\bf Pr}(w\in \Lan(\A^{\rightarrow p}))=(1-p){\bf 1}_{w\in\Lan(\A)}+p{\bf 1}_{w \not\in\Lan(\A)}$$ 
where ${\bf 1}_{C}$ is 1 if condition $C$ holds, and 0 otherwise.
Observe that the expected value of the distance  
$d(\Lan(\A),\Lan(\A^{\rightarrow p}))$ is $p$. Moreover, in our experiments, 
we observe that $\left|\frac{d(\Lan(\A),\Lan(\A^{\rightarrow p}))-p}{p}\right|<5\cdot 10^{-2}$ for all the generated languages.

\paragraph{DFA with noisy input.}

Given a DFA $\A$ over the alphabet $\Sigma$ (with $|\Sigma|>1$) and $0<p<1$, the random language
$\Lan(\A^{\leftarrow p})$ 
changes every letter of the word with probability $p$
uniformly to another letter and then returns the classification of the new word w.r.t.\ $\Lan(\A)$. More formally, for $w=a_1\ldots a_n \in \Sigma^\ast$,
$${\bf Pr}(w\in \Lan(\A^{\leftarrow p}))=\sum_{\substack{w'=b_1\ldots b_n\in\Lan(\A) \\\textup{s.t. }|w|=|w'|}}~\prod_{1 \le i\leq n} \Bigl((1-p){\bf 1}_{a_i=b_i}+\frac{p}{|\Sigma|-1} {\bf 1}_{a_i\neq b_i}\Bigr)\,.$$

{\noindent \bf Counter DFA.}
Let $\A$  be a DFA over the alphabet $\Sigma$ and $c: \Sigma \cup \{\lambda\} \to \integer$ be a function.
We inductively define the function $\overline{c}: \Sigma^* \to \integer$ by
$$ 
	\overline{c}(\lambda)= c(\lambda) \mbox{ and } \overline{c}(wa)= \overline{c}(w)+c(a)\,.
$$
The counter language $\Lan(\A_c)$ is now given as
$$\Lan(\A_c) = \Lan(\A) \cup \{w \in \Sigma^\ast \mid \overline{c}(w)\leq 0\}\,.$$

%% file: evaluation.tex
\section{Experimental Evaluation}
\label{sec:evaluation}

In order to empirically evaluate our ideas, we have implemented a prototype and benchmarks in Python, using the NumPy library. 
They are available on GitHub\footnote{\url{https://github.com/LeaRNNify/Noisy_Learning}}.
All evaluations were performed on a computer equipped by Intel i5-8250U CPU with 4 cores, 16GB of memory and Ubuntu Linux~18.03.

\subsection{Generating DFAs}

We now describe the settings of the experiments we made with three different types of noises. 
We choose $\mu=10^{-2}$ for the parameter of the word distribution so that the average length
of a random word is $99$.
All the statistic distances were computed using the Chernoff-Hoeffding bound~\cite{Hoff63} 
with $\alpha = 5\cdot 10^{-4}$ as error parameter and $\gamma = 10^{-3}$ as confidence level. 

The benchmarks were performed on  DFA randomly generated using the following procedure.
Let $M_q=50$ and $M_a=20$ be two parameters, which impose upper bounds on the number of states and of the alphabet, that could be tuned in future experiments. 
The DFA $\A=(Q,\sigma,q_0,F)$ on $\Sigma$ is generated as follows:
\begin{itemize}
	\item Uniformly choose $n_q\in[10,M_q]$ and $n_a\in[3,M_a]$;
	\item Set $Q = [0,n_q]$ and $\Sigma = [0,n_a]$;
	\item Uniformly choose $n_f\in[0,n_q-1]$ and let $F = [0,n_f]$;
	\item Uniformly choose $q_0$ in $Q$;
	\item For all $(q,a)\in Q\times\Sigma$, choose the target state $\sigma(q,a)$ uniformly among all states.
\end{itemize}
The choice of $M_q$ and $M_a$ was inspired by observing that these values often occur when
 modeling  realistic processes like in business process 
management.

\subsection{Tunings} \label{subsec:tunings}

Before launching our experiments, we first tune two key parameters for both efficiency and accuracy purposes:
the maximal number of rounds of the algorithm and the accuracy of the approximate equivalence query.
This tuning is based on experiments over the DFA with the noisy output since the expected distance between the DFA and the noisy device is known ($p$),
thus simplifying the tuning.

\paragraph{Maximal number of rounds.}

In order to specify a maximal number of rounds that lead to the good performances of the Angluin's Algorithm, we took a DFA with noisy output 
$\A^{\rightarrow p}$ for $p\in\{0.005,0.0025,0.0015,0.001\}$.
We ran the learning algorithm, stopping every 20 rounds to estimate the distance between the current DFA $\A_E$ to the original DFA $\A$.
Figure~\ref{fig:roundAnalysis1} shows the evolution graphs of  $d(\Lan(\A),\Lan(\A_E))$ w.r.t. the number of rounds 
according to the different values of $p$ each of them summarizing five runs on five different DFAs. 
The vertical axis corresponds to the distance to original DFA $\A$, and the horizontal axis corresponds to  the number of rounds. 
The red line is the distance with $\A^{\rightarrow p}$, and the blue line is the distance with $\A_E$.
\begin{figure}[h]
	\centering
	\includegraphics[scale=0.36]{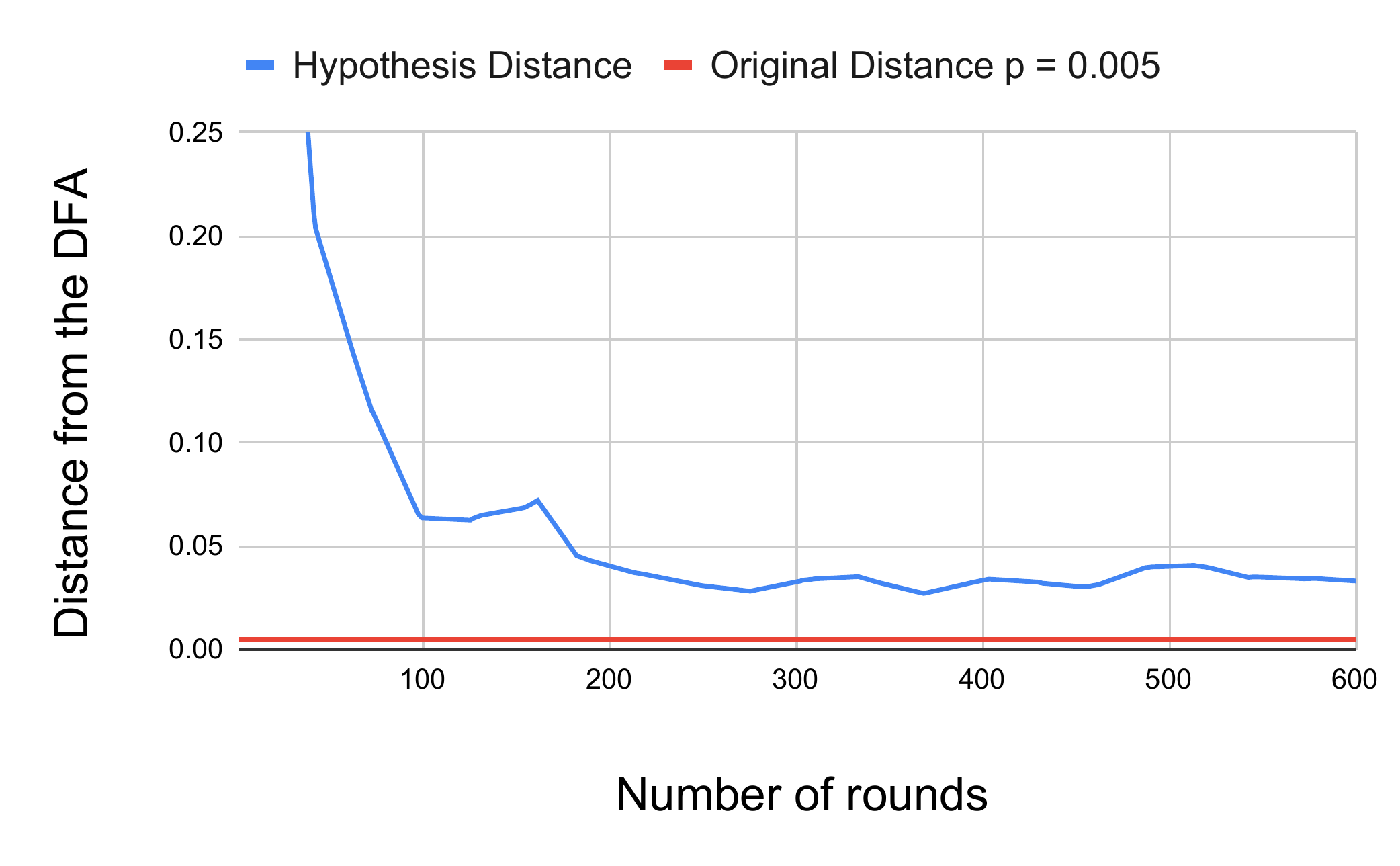}
	\includegraphics[scale=0.36]{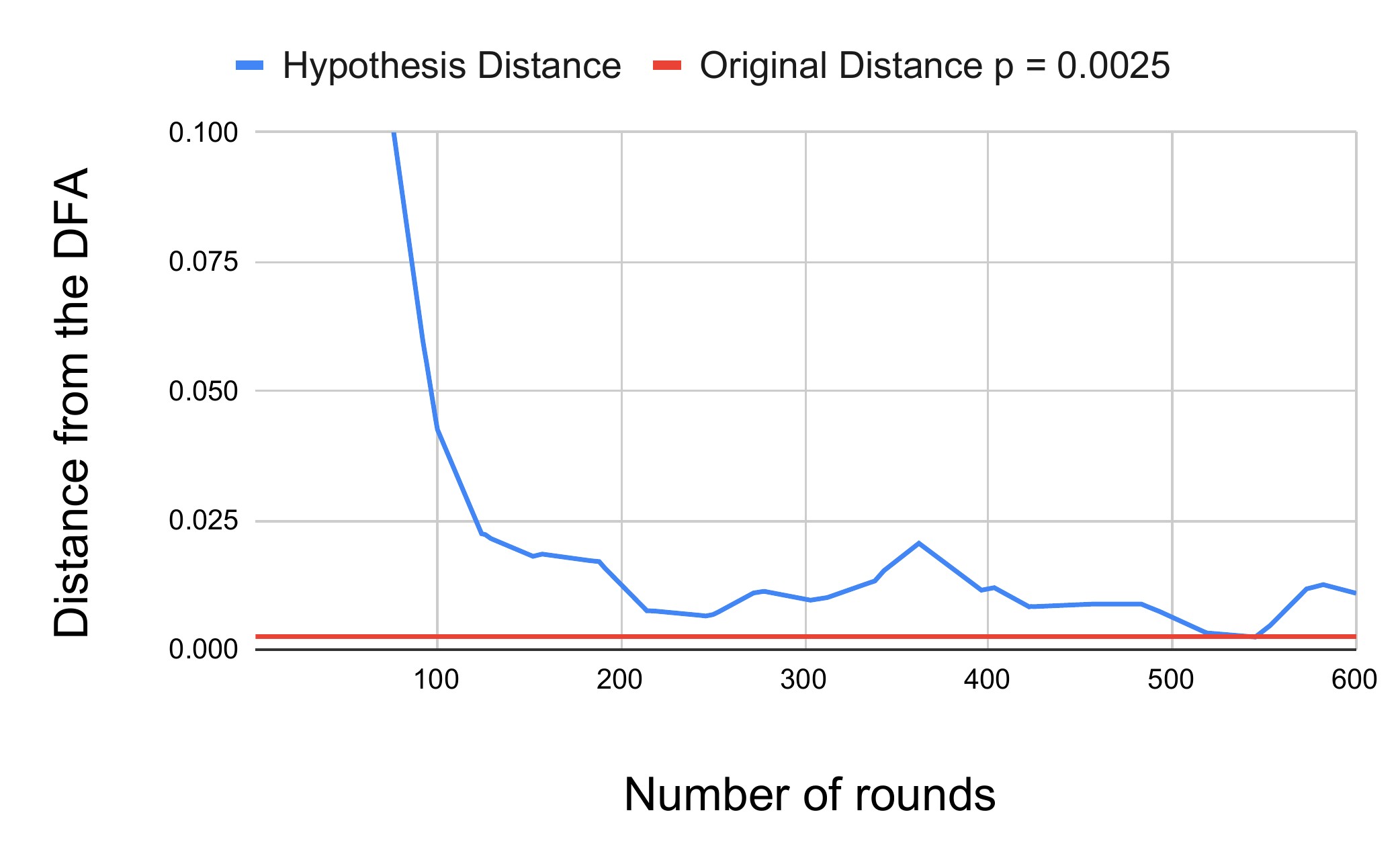}
	\includegraphics[scale=0.36]{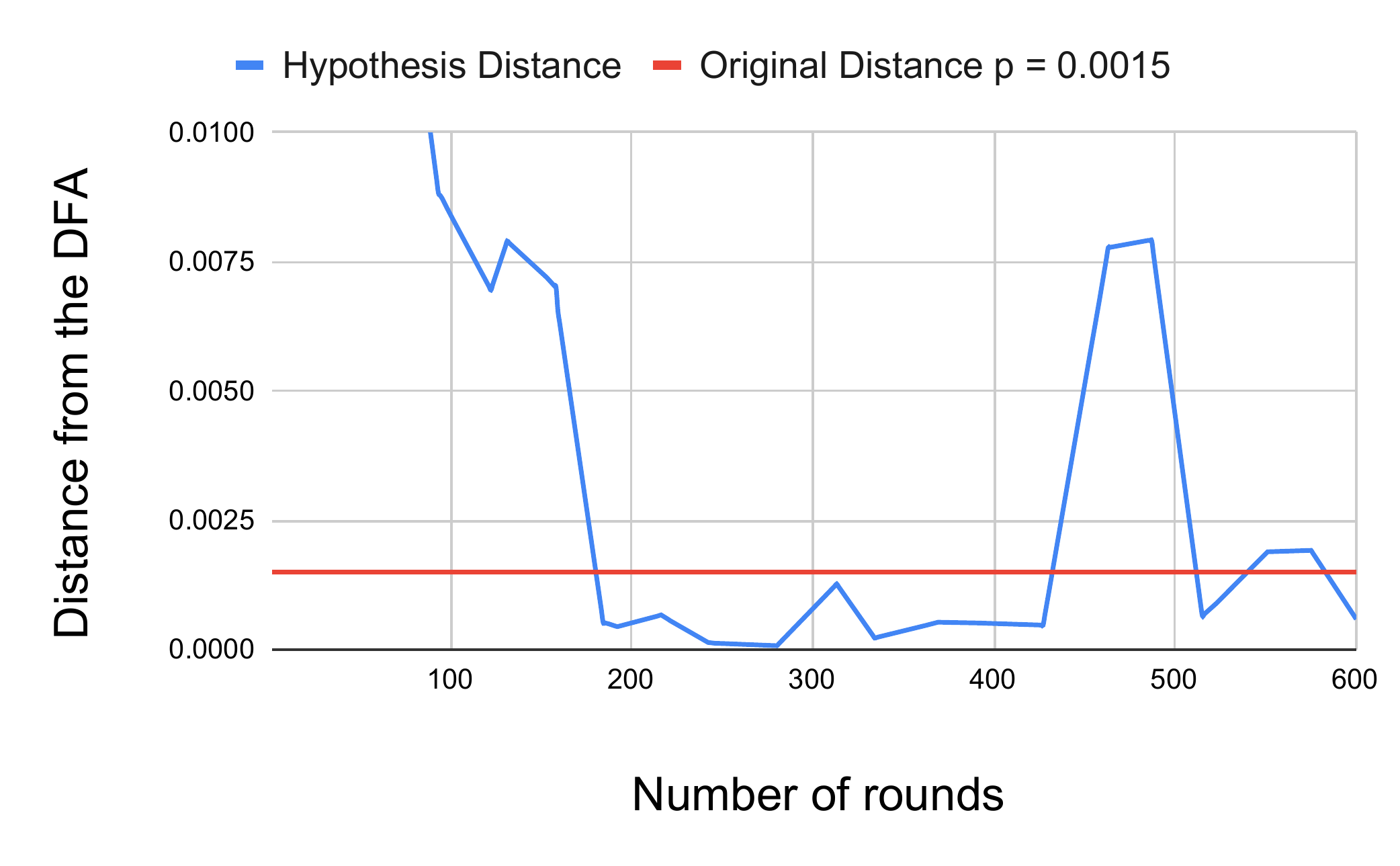}
	\includegraphics[scale=0.36]{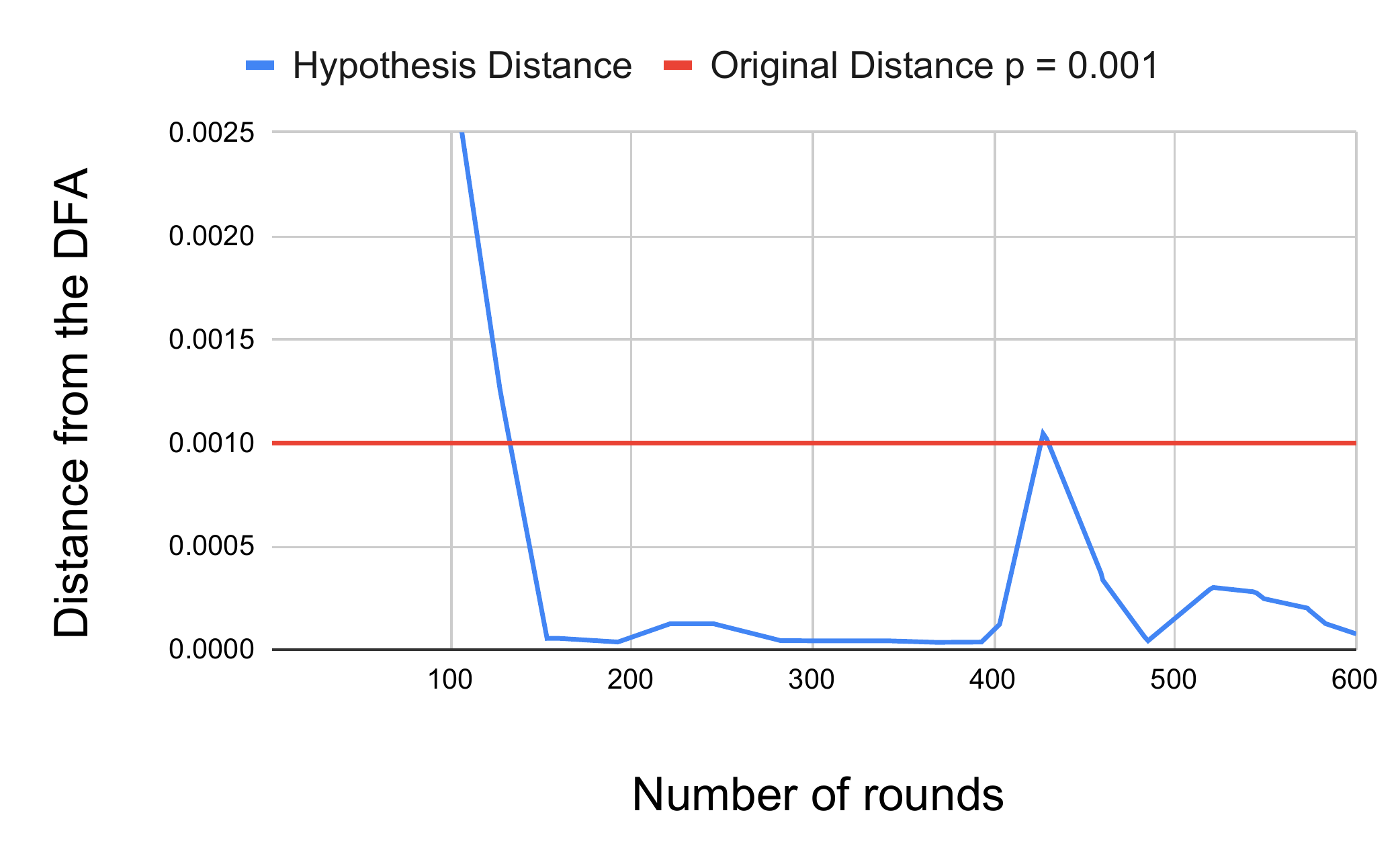}
	\caption{Number of rounds analysis}
	\label{fig:roundAnalysis1}
\end{figure}
We observe that after about 250 rounds $d(\Lan(\A),\Lan(\A_E))$ is stabilizing except some rare peaks, which are worth further investigation. 
Therefore, from now on all the experiments are made with a maximum of 250 rounds.
Of course this number depends on the size of $\A$ but for the variable size that we have chosen (between 10 and 50 states)
it seems to be a good choice. 

\paragraph{Accuracy of the approximate equivalence query.}
We have generated thirty-five DFA and for each of them we generated five $\A^{\rightarrow p}$ with different values of $p$.
Table~\ref{tbl:epsilonAndDelatEpirementet} summarizes our results with different $\epsilon$ and $\delta$ for the approximate equivalence query. 
The rows correspond to the value of the noise $p$, the columns correspond to the values of $ \epsilon $ and $ \delta $ (where we always choose $\epsilon=\delta$) and each cell shows the average information gain. Looking at this table, 
$\varepsilon=\delta=0.01$ and $\varepsilon=\delta=0.005$ seem to be optimal values. We decided to fix  $\varepsilon=\delta=0.005$
for all our experiments.

	\begin{table}
		\centering
		\scalebox{.8}{
		\begin{tabular}{l*{5}{@{\hskip1em}r}}
			\toprule
			\backslashbox{$p$}{$\varepsilon=\delta$} & $0.05$ & $0.01$ & $0.005$ & $0.001$ & $0.0005$
			\tabularnewline
			\midrule
			$0.01$ &\cellcolor{low}$0.081$ & \cellcolor{low}$0.054$ & \cellcolor{low}$0.047$ & \cellcolor{low}$0.048$ & \cellcolor{low}$0.050$ \tabularnewline
			$0.005$ & \cellcolor{low}$0.086$ & \cellcolor{low}$0.087$ & \cellcolor{low}$0.072$ & \cellcolor{low}$0.070$ & \cellcolor{low}$0.094$ \tabularnewline
			$0.0025$ & \cellcolor{low}$0.867$ & \cellcolor{low}$0.292$ & \cellcolor{low}$0.591$ &\cellcolor{low} $0.321$ & \cellcolor{low}$0.748$ \tabularnewline
			$0.0015$ & \cellcolor{medium}$1.401$ & \cellcolor{high}$2.933$ & \cellcolor{high}$3.082$ & \cellcolor{medium}$0.980$ & \cellcolor{low}$0.710$ \tabularnewline
			$0.001$ & \cellcolor{high}$5.334$ & \cellcolor{high}$4.524$ & \cellcolor{high}$3.594$ & \cellcolor{high}$1.811$ & \cellcolor{high}$6.440$ \tabularnewline
			\bottomrule
		\end{tabular}
		}
		\caption{ Evaluation of the impact of $\epsilon$ and $\delta$. }\label{tbl:epsilonAndDelatEpirementet}
	\end{table}

\subsection{Qualitative and Quantitative analysis}
For the three types of noise we have generated numerous DFA (as described shortly above), and for each DFA we have generated several noisy devices depending on the `quantity' of noise.
By computing the (average) information gain for all these experiments, we have been able to get conclusions about the effect of the nature and 
the quantity of the noise on the performance of Angluin's algorithm.

When Angluin's algorithm is applied to a noisy device, a corresponding
random language is generated on-the-fly: once membership of a word in the
target language has been determined (e.g., through a membership query),
the corresponding truth value is stored and not changed anymore.

\paragraph{DFA with noisy output.}
We have generated fifty DFA, and for each such DFA $\A$, we have generated random languages
with noisy output $\Lan(\A^{\rightarrow p})$ with five values for $p$ between $0.01$ and $0.001$.
Table~\ref{tbl:noisyoutput} summarizes the results.
Recall that the expected value of $d(\mathcal L(\A),\mathcal L(\A^{\rightarrow p}))$ is $p$.
We have identified a threshold for $p$ between $0.0015$ and $0.0025$:
if the noise is above $0.0025$ the resulting DFA $\mathcal{A}_E$ has a bigger distance to the original one $\A$ than $\A^{\rightarrow p}$, and smaller if the noise is under $0.0015$.
Moreover, once we cross the threshold the robustness of the algorithm increases very quickly. We have also included a column
that represents the standard deviation of the random variable $d(\mathcal L(\A),\mathcal L(\A_E))$ to assess that our conclusions are robust w.r.t.
the probabilistic feature.

	\begin{table}[h]
		\centering
		\label{tbl:summaryBenchmarkNoisy}
		\scalebox{.8}{
		\begin{tabular}{l*{4}{r}}
			\toprule
			$ p $ &$d(\mathcal L(\A),\mathcal L(\A_E))$ & $d(\mathcal L(\A^{\rightarrow p}),\mathcal L(\A_E))$ &gain  & standard deviation\tabularnewline
			\midrule
			$0.01$ & $0.12625$ & $0.13320$ &\cellcolor{low} $0.07432$ & $0.04102$	\tabularnewline
			$0.005$ & $0.04420$ & $0.04827$ & \cellcolor{low} $0.11312$  & $0.03366$	\tabularnewline
			$0.0025$ & $0.00333$ & $0.00568$ &\cellcolor{low} $0.75031$ & $0.00523$ \tabularnewline
			$0.0015$ & $0.00027$ & $0.00174$ &\cellcolor{high} $5.52999$ & $0.00047$ \tabularnewline
			$0.001$ &  $0.00006$ & $0.00103$ &\cellcolor{high} $15.75817$& $0.00007$ \tabularnewline
			\bottomrule
		\end{tabular}
	}
		
		\caption{Evaluation of the algorithm w.r.t.\ the noisy output.}\label{tbl:noisyoutput}
	\end{table}

\paragraph{DFA with noisy input.}
We have generated forty-five random DFA, and for each such DFA $\A$,
we have generated random languages $\Lan(\A^{\leftarrow p})$ with $p\in\{10^{-4},5\cdot10^{-4},10^{-3},5\cdot10^{-3}\}$. 
Contrary to the case of noisy output, $p$ does not correspond to the expected value of $d(\mathcal L(\A),\mathcal L(\A^{\leftarrow p})$.
Thus we have evaluated this distance for every pair of the experiments and we have gathered the pairs whose distances belong 
to intervals that are described in the first column of Table~\ref{tbl:summaryBenchmarkNoisyInput}.
The second column of this table reports the number of pairs in the interval while the third one reports 
the average value of this distance for these pairs.
Again we identify a threshold for $d(\mathcal L(\A),\mathcal L(\A^{\leftarrow p}))$ between $0.001$ and $0.005$ and
once we cross the threshold the robustness of the algorithm increases very quickly.

	\begin{table}[h]
		\centering

		\scalebox{.8}{
		\begin{tabular}{l*{6}{r}}
			\toprule
			Range & \# &$d(\mathcal L(\A),\mathcal L(\A^{\leftarrow p}))$  &$d(\mathcal L(\A),\mathcal L(\A_E))$ &$d(\mathcal L\A^{\leftarrow p}),\mathcal L(\A_E))$ &gain& standard deviation\tabularnewline
			\midrule 
			$[0.025, 1]$ & $36$ & $0.04027$ & $0.21513$ & $0.22658$ & \cellcolor{low}$0.18$& 0.05279\tabularnewline
			$[0.005, 0.025]$ & $53$ & $0.00924$ & $0.05416$ & $0.06077$ &\cellcolor{low} $0.17$ & 0.04172\tabularnewline
			$[0.002,0.005]$ &33& $0.00378$ & $0.01260$ & $0.01611$ &\cellcolor{low} $0.30$ & 0.01783\tabularnewline
			$[0.001,0.002]$ & $11$ & $0.00123$ & $0.00030$ & $0.00154$ &\cellcolor{high} $4.1$ & 0.00058\tabularnewline
			$[0.0005,0.001]$ & $25$ & $0.00079$& $0.00002$ & $0.00082$&\cellcolor{high} $39.5$ & 0.00007\tabularnewline	
			\bottomrule
		\end{tabular}
	}

		\caption{Evaluation of the algorithm w.r.t.\ the noisy input.}\label{tbl:summaryBenchmarkNoisyInput}
	\end{table}

\paragraph{Counter DFA.}
We have randomly generated the counter function as follows:
We have uniformly chosen $c(\lambda)$ in $[0,|\Sigma|]$. Then,
for all $a \in \Sigma$, ${\bf Pr}(c(a)=-1)=\frac 1 4$ and for all $0\leq i \leq 6$, ${\bf Pr}(c(a)=i)=\frac 3 {28}$.

We have generated 160 DFA. For each of them, we have generated a counter automaton (as described before). 
The results of our experiments are given in Table~\ref{tbl:counterDFA}. Here whatever the quantity of noise the Angluin's algorithm is unable to
get closer to the original DFA. Moreover the extracted DFA $\A_E$ is very often closer to the counter automaton $\A_c$ than the original DFA $\A$. 
	
	\begin{table}[h]
		\centering
		\scalebox{.8}{
		\begin{tabular}{l*{6}{r}}
			\toprule
			Range& \# & $d(\mathcal L(\A),\mathcal L(\A_{c}))$  &$d(\mathcal L(\A),\mathcal L(\A_E))$ & 
			$d(\mathcal L(\A_c),\mathcal L(\A_E))$ &gain& standard deviation \tabularnewline
			\midrule
			$[0.005, 0.025]$ & $14$ & $0.01238$ & $0.02586$ & $0.02053$ & \cellcolor{low}$0.47886$ & $0.01898$ \tabularnewline
			$[0.002, 0.005]$ & $57$ & $0.00245$ & $0.00396$ & $0.00262$ & \cellcolor{low}$0.61765$ & $0.00298$ \tabularnewline
			$[0.001, 0.002]$ & $22$ & $0.00143$ & $0.00209$ & $0.00121$ & \cellcolor{low}$0.68156$ & $0.00126$ \tabularnewline
			$[0.0005, 0.001]$ & $20$ & $0.00079$ & $0.00108$ & $0.00064$ & \cellcolor{low}$0.72481$ & $0.00065$ \tabularnewline
			$[0.0001, 0.0005]$ & $44$ & $0.00025$ & $0.00035$ & $0.00021$ & \cellcolor{low}$0.71054$ & $0.00021$ \tabularnewline
			\bottomrule
		\end{tabular}
			}

		\caption{Evaluation of the algorithm w.r.t.\ the `noisy' counter.}\label{tbl:counterDFA}
	\end{table}

Thus we conjecture that when the noise is `unstructured' and the quantity is small enough such that the word noise
is still meaningful, then Angluin's algorithm is robust. On the contrary when the noise is structured, then Angluin's algorithm
`tries to learn' the noisy device whatever the quantity of noise. In Section~\ref{sec:theory}, we will strengthen this conjecture
establishing that in some sense noise produced by random process implies unstructured noise. 

\subsection{Words distribution }\label{subsection:wordDist}
We now discuss the impact of word distribution on the robustness of the Angluin algorithm.
The parameter $\mu$ determines the average length of a random word ($\frac{1}{\mu}-1$).
Table~\ref{tbl:wordDistrabution} summarizes experimental results with values of $\mu$ indicated on the first row. 
The other rows correspond to different values of the noise $p$ for $\A^{\rightarrow p}$. The cells (at the intersection of a pair ($p$,$\mu$)) contain the (average) information gain where experiments have been done over twenty-two DFA always eliminating the worst and best cases to avoid that the pathological cases perturb the average values. For values of $p$ that matter (i.e., when the gain is greater than 1), there is clear tendency for the gain to first increase w.r.t. $\mu$, reaching a maximum
about $\mu=0.01$ the value that we have chosen and then decrease. A possible explanation would be the following: too short words (i.e., big $\mu$) does not help
to discriminate between languages while  too long words (i.e., small $\mu$) lead to overfitting and does not reduce the noise.

\begin{table}[h]
	\centering
	\scalebox{.8}{
	\begin{tabular}{l*{5}{@{\hskip1em}r}}
		\toprule
		\backslashbox{$p$}{$\mu$} & $0.001$ & $0.005$ & $0.01$ & $0.05$ & $0.1$ \tabularnewline
		\midrule
		$0.01$  &\cellcolor{low}$0.059$ & \cellcolor{low}$0.067$ & \cellcolor{low}$0.078$ & \cellcolor{low}$0.184$ & \cellcolor{low}$0.317$ \tabularnewline
		$0.005$ & \cellcolor{low}$0.078$ & \cellcolor{low}$0.130$ & \cellcolor{low}$0.134$ & \cellcolor{low}$0.559$ & \cellcolor{low}$0.966$ \tabularnewline
		$0.0025$ & \cellcolor{low}$0.165$ & \cellcolor{low}$0.298$ & \cellcolor{low}$0.398$ &\cellcolor{medium}$1.246$ & \cellcolor{low}$0.823$ \tabularnewline
		$0.0015$ & \cellcolor{low}$0.465$ & \cellcolor{low}$0.671$ & \cellcolor{high}$2.267$ & \cellcolor{high}$2.074$ & \cellcolor{high}$1.651$ \tabularnewline
		$0.001$ & \cellcolor{high}$1.801$ & \cellcolor{high}$10.94$ & \cellcolor{high}$8.907$ & \cellcolor{high}$3.753$ & \cellcolor{high}$2.341$ \tabularnewline
		\bottomrule 
	\end{tabular}
}
	\caption{Analysis of different distributions on $\Sigma^*$ }\label{tbl:wordDistrabution}
\end{table}

%% file: theory.tex
 
\section{Random languages versus structured languages}
\label{sec:theory}

Recall that in the precedent section, from the experimental results, we conjecture that Angluin’s algorithm is robust, when the noise is random, i.e., unstructured, and its quantity is small enough, such as for DFA with noisy output and with noisy input. This is however not the case for structured counter DFA, for which Angluin’s algorithm learns the noisy device itself instead of the original one whatever the quantity of noise.

In this section, we want to theoretically establish  that the main factor of the robustness
of the Angluin's L$^*$ algorithm w.r.t.\ random noise is that almost surely randomness, in most cases, 
yields the perturbated language that is unstructured. We consider a language as structured if it can
be produced by some general device. Thus we identify the family of structured languages
with the family of recursively enumerable languages. More precisely, we show that almost surely DFA with noisy output leads to a language that is not recursively enumerable. We then demonstrate further that with a mind condition, almost surely DFA with noisy input yields also non-recursively enumerable language. As for the counter DFA, by definition, it is clearly recursively enumerable, thus not being studied further. 

The following lemma gives a simple means to establish that almost surely 
a random language is not recursively enumerable.

\begin{lemma}
\label{lemma:nonre}
Let $R$ be a random language over $\Sigma$. Let $(w_n)_{n\in \nat}$ be a sequence of words of $\Sigma^*$.
Let $W_n=\{w_i\}_{i<n}$ and $\rho_n=\max_{W\subseteq W_n}{\bf Pr}(R \cap W_n=W)$.
Assume that $\lim_{n\rightarrow \infty} \rho_n=0$. 
Then, for all countable families of languages $\mathcal F$, almost surely $R \notin \mathcal F$.
In particular, almost surely $R$ is not a recursively enumerable language.
\end{lemma}
\begin{proof}
Let us consider an arbitrary language $L$.
Then, for all $n$, $ {\bf Pr}(R = L) \leq {\bf Pr}(R \cap W_n=L \cap W_n)\leq\rho_n$.
Thus, ${\bf Pr}(R =L)=0$ and ${\bf Pr}(R \in \mathcal F)=\sum_{L\in \mathcal F}{\bf Pr}(R =L)=0$.
\qed
\end{proof}

From Lemma~\ref{lemma:nonre}, we immediately obtain that almost surely
the noisy output perturbation of any language is not recursively enumerable.
The proofs of the two next theorems use the same notations as those given in Lemma~\ref{lemma:nonre}.
\begin{theorem}
\label{theorem:un}
	Let $L$ be a  language and $0<p<1$. Then almost surely $L^{\rightarrow p}$ is not a recursively enumerable language.
\end{theorem}
\begin{proof}
Consider any enumeration  $(w_n)_{n\in \nat}$ of $\Sigma^*$ and any $W\subseteq W_n$.
The probability that $L^{\rightarrow p}\cap W_n$ is equal to $W$  is bounded
by $\max(p,1-p)^n$. Thus, $\rho_n\leq \max(p,1-p)^n$ and $\lim_{n\rightarrow \infty} \rho_n=0$.
\qed
\end{proof}

We cannot get a similar result for the noisy input perturbation. Indeed consider the language
$\Sigma^*$, whatever the kind of noise brought to the input, the obtained language is still $\Sigma^*$.
With the kind of input noise that we study, consider the language that accepts words of odd length 
(see the automaton $\A'$ of Figure~\ref{fig:inOutProp2}).
Then the perturbed language is unchanged. 

However given a DFA $\mathcal A$, we establish a mild condition
on $\mathcal A$ ensuring that almost surely the random language $\Lan(\mathcal A^{\leftarrow p})$ is not recursively enumerable.
We abbreviate bottom strongly connected components (of $\mathcal A$ viewed as a graph) by BSCC.

\begin{definition}
Let $\A = (Q,F,\sigma,q_0)$ be a DFA.
We call $\A$ \emph{equal-length-distinguishing} if
there exist (possibly identical) BSCC $\mathcal C, \mathcal C'$ of $\A$,
$q_1\in \mathcal C \cap F$, $q'_1\in \mathcal C' \setminus F$, and $w,w'\in \Sigma^*$ such that
we have $q_1 = \sigma(q_0,w)$, $q'_1 = \sigma(q_0,w')$, and $|w|=|w'|$.
\end{definition}

\begin{theorem}
\label{theorem:deux}
	Let $\Sigma$ be an alphabet with $|\Sigma|>1$.
	Let $\A = (Q,F,\sigma,q_0)$ be a DFA over $\Sigma$, $0<p<1$ and $\mathcal C, \mathcal C'$ some BSCC
	of $\A$ (possibly equal).
	Assume that $\A$ is equal-length-distinguishing.
	Then almost surely $\Lan(\A^{\leftarrow p})$ is not a recursively enumerable language.
\end{theorem}
\begin{proof}
Let us denote  $\ell=|w|$ and $m$ (resp. $m'$) the periodicity of $\mathcal C$ (resp. $\mathcal C'$).
Moreover, let $a\in \Sigma$.
We build a Markov chain $\mathcal M$ from $\mathcal C$ as follows: every transition
	$q\xrightarrow{a}q'$ has probability $1-p$ and for all $b\neq a$, every transition $q\xrightarrow{b}q'$ has probability $\frac{p}{|\Sigma|-1}$.
	We proceed similarly from $\mathcal C'$ to build $\mathcal M'$.

Let us denote $\alpha_n$ (resp.  $\alpha'_n$ ) the probability in $\mathcal M$ (resp. $\mathcal M'$) that starting from $q_1$ 
	(resp.  $q'_1$), the current state at time $n$ is $q_1$ (resp. $q'_1$).
        Since $\mathcal M$ and $\mathcal M'$ are irreducible, $\lim_{n\rightarrow \infty} \alpha_{mn}$ (resp. $\lim_{n\rightarrow \infty} \alpha'_{m'n}$)
         exists and is positive. Let us denote $\alpha$ (resp. $\alpha'$) this limit. There exists $n_0$ such that for all $n\geq n_0$, 
         $\alpha_{mn}\geq \frac{\alpha}{2}$ and $ \alpha'_{m'n} \geq \frac{\alpha'}{2}$.
         
Define $w_n=wa^{mm'(n+n_0)}$ for all $n\in \N$. 
The probability that $w_n$ is accepted by $\Lan(\A)^{\leftarrow p}$ is lower bounded by  the probability 
         that the prefix $w$ is unchanged (thus reaching $q_1$) and that after $mm'(n+n_0)$ steps the current state in $\mathcal M$ is $q_1$.
         So a lower bound is: $\min(p,1-p)^\ell \frac{\alpha}{2}$.

The probability that $w_n$ is rejected by $\Lan(\A)^{\leftarrow p}$ is lower bounded by  the probability 
         that the prefix $w$ is changed into $w'$ (thus reaching $q'_1$) and that after $mm'(n+n_0)$ steps the current state in $\mathcal M'$ is $q'_1$.
         So a lower bound is: $\min(p,1-p)^\ell \frac{\alpha'}{2}$.
		
Let $W\subseteq W_n$. The probability that $L^{\leftarrow p}\cap W_n$ is equal to $W$  is upper bounded
	by: $$\left(1-\min(p,1-p)^\ell \frac{\min(\alpha,\alpha')}{2}\right)^n$$ 
	
Thus $\rho_n\leq \left(1-\min(p,1-p)^\ell \frac{\min(\alpha,\alpha')}{2}\right)^n$ and $\lim_{n\rightarrow \infty} \rho_n=0$.	
	\qed
\end{proof}

The DFA $\A$ of Figure~\ref{fig:inOutProp2} that represents the formula `$a \text{ Until } b$' of temporal logic LTL
is equal-length-distinguishing. The corresponding pair of states consists of the accepting state and the leftmost one.
Checking the hypotheses of this theorem can be done in quadratic time by first building a graph whose set of vertices is $Q\times Q$
and there is an edge $(q_1,q_2) \rightarrow (q'_1,q'_2)$ if there are some transitions $q_1 \xrightarrow{a_1} q'_1$ and 
$q_2 \xrightarrow{a_2} q'_2$ and then looking for a vertex $(q_1,q_2)$ in some BSCC with $q_1\in F$
and $q_2\notin F$ reachable from $(q_0,q_0)$.

\begin{figure}[h]
	\centering
	
	\includegraphics[scale=0.87]{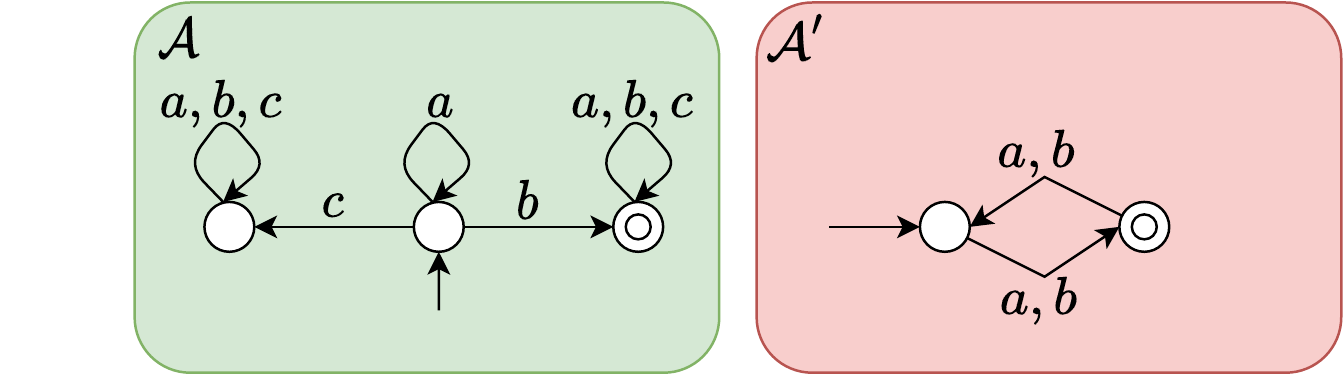}
	\caption{Two DFA}
	\label{fig:inOutProp2}
\end{figure}

We have that the property of being equal-length-distinguishing is a sufficient condition for ensuring that almost surely $\Lan(\A^{\leftarrow p})$ is not a recursively enumerable language. So we want to investigate whether it is a necessary condition.
The next proposition shows a particular case when this condition is necessary.

\begin{proposition}
	Let $\Sigma$ be an alphabet with $|\Sigma|>1$.
	Let $\A = (Q,F,\sigma,q_0)$ be a DFA that is not equal-length-distinguishing and such that every circuit of $\A$ belongs to a BSCC.
	Then, for every sampling $L'$ of $\Lan(\A^{\leftarrow p})$, $L'$ is regular.
\end{proposition}
\begin{proof}
Pick some $n_0\in \N$ such that for all $w$ with $|w|\geq n_0$ and $q_0\xrightarrow{w} q$ implies that $q$ belongs to some BSCC.
Observe now that, since $\A$ is not equal-length-distinguishing, for  words $w,w'$ with $|w|=|w'|\geq n_0$,
$w\in L$ iff $w'\in L$. Thus, for every sampling $L'$ of $\Lan(\A^{\leftarrow p})$, $L'=(L' \cap \Sigma^{<n_0}) \cup (L\cap \Sigma^{\geq n_0})$
implying that $L'$ is regular.
\qed
\end{proof}
\begin{figure}[h]
\begin{footnotesize}
\begin{center}
\begin{tikzpicture}[xscale=0.9,yscale=0.99]
\begin{scope}[auto,every node/.style={draw,circle,minimum size=20pt}]

\path (3,0) node [draw,circle,inner sep=2pt](l3) {$q_0$};
\path (9,0) node [draw,circle,double,inner sep=2pt](lf) {$q_f$};
\path (6,0) node [draw,circle,inner sep=2pt](lr) {$q_r$};

\end{scope}
 
\draw[arrows=-latex'] (2,0) -- (l3);
\draw[arrows=-latex'] (l3) --(6,1.2)--(lf)node[pos=0,above] {$a$};
\draw[arrows=-latex'] (lf) --(6,-1.5)--(l3)node[pos=0,below] {$b$};
\draw[-latex'] (lf) .. controls +(-300:40pt) and +(-240:40pt) .. (lf)
node[pos=0.5,above] {$a$};
\draw[-latex'] (l3) .. controls +(-300:40pt) and +(-240:40pt) .. (l3)
node[pos=0.5,above] {$b$};
\draw[arrows=-latex'] (l3) --(lr)node[pos=0.5,above] {$c$};
\draw[arrows=-latex'] (lf) -- (lr)node[pos=0.5,above] {$c$};
\draw[-latex'] (lr) .. controls +(300:40pt) and +(240:40pt) .. (lr)
node[pos=0.5,above] {$\Sigma$};

\end{tikzpicture}
\end{center}
\end{footnotesize}
\caption{A DFA $\A$ with $\Lan(\A)=(a+b)^*a$}
\label{fig:counter-example}
\end{figure}
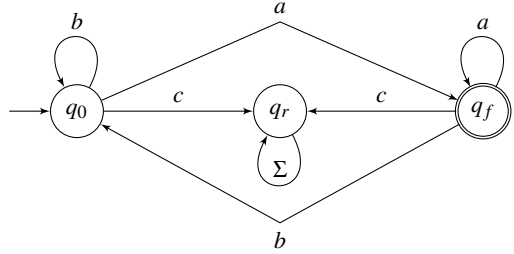
Observe that we establish the next proposition using a generalization of Lemma~\ref{lemma:nonre}.
\begin{proposition}
Let $\A$ be the DFA of Figure~\ref{fig:counter-example}.
Then, $\A$ is not equal-length-distinguishing. Moreover, almost surely $\Lan(\A^{\leftarrow \frac 2 3})$ is not recursively enumerable.
\end{proposition}
\begin{proof}
There is a single BSCC  with a single state $\{q_r\}$. So $\A$ is not equal-length-distinguishing.
Let $w\neq \lambda$ be a word  with $|w|=n$ and denote $\tilde{w}$ the random word obtained by the noisy perturbation. 
Observe that every letter of $\tilde{w}$ is uniformly distributed over $\Sigma$.
So the probability that $\tilde{w}$ does not contain a $c$ is $(\frac 2 3)^n$ and the conditional probability that $\tilde{w}$
belongs to $\Lan(\A^{\leftarrow \frac 2 3})$ knowing that it does not contain a $c$ is $\frac 1 2$.

\noindent
Fix some $0<\rho<1$. The probability that for all words $w\in \Sigma^n$, $\tilde{w}$ 
contains a $c$ is equal to $(1-(\frac 2 3)^n)^{3^n}\leq e^{-2^n}$. Pick an increasing sequence $(n_k)_{k\in \N}$ such $\sum_{k\in \N} e^{-2^{n_k}}\leq 1-\rho$.
Then with probability at least $\rho$, for all $k$, there is a word  $w_k\in \Sigma^{n_k}$ such that $\tilde{w}_k$ does not contain a $c$. Letting $\rho$
go to 1, almost surely there is an infinite number of words $w$ such that  $\tilde{w}\in (a+b)^+$. 

\noindent
Let us consider an arbitrary language $L'$ and $(w_n)_{n\in \N}$ be an enumeration of $\Sigma^+$.
Then almost surely there is an infinite number of $w_n$ such that $\tilde{w}_n$ belong to $(a+b)^+$. 
Recall that for such a word, the probability that it belongs to $\Lan(\A^{\leftarrow \frac 2 3})$ is equal to $\frac 1 2$.
Let $W_n$ be the \emph{random} set of the first $n^{th}$ such words.
Then for all $n$, $ {\bf Pr}(L' =\Lan(\A^{\leftarrow \frac 2 3})) \leq {\bf Pr}(L' \cap W_n=\Lan(\A^{\leftarrow \frac 2 3}) \cap W_n)=2^{-n}$.

Thus ${\bf Pr}(L' =\Lan(\A^{\leftarrow \frac 2 3}))=0$ and ${\bf Pr}(\Lan(\A^{\leftarrow \frac 2 3}) \in \mathcal F)=\sum_{L'\in \mathcal F}{\bf Pr}(L' =\Lan(\A^{\leftarrow \frac 2 3}))=0$
for $ \mathcal F$ a countable family of languages.
\qed
\end{proof}

To show the soundness of the structural criterion in Theorem~\ref{theorem:deux} with experiments and comparisons, we have refined our experiments on DFA with noisy inputs partitioning the randomly generated DFA depending on whether they are equal-length-distinguishing. 

We have chosen $|\Sigma| = 3$ since with greater size, it was difficult to generate DFAs that do not satisfy the hypotheses.
Tables~\ref{tbl:noisyWoSpecialProperty} and~\ref{tbl:noisySpecialProperty} summarize these experiments.
The last rows of the tables (where the information gain is greater than one) confirm our conjecture.

\begin{table}[h]
	\centering
	\scalebox{.8}{
	\begin{tabular}{l*{5}{@{\hskip1em}r}}
		\toprule
		Range& \# & $d(\mathcal L(\A),\mathcal L(\A^{\leftarrow p}))$  & $d(\mathcal L(\A),\mathcal L(\A_E))$ &  $d(\mathcal L(\A^{\leftarrow p})),\mathcal L(\A_E))$ &gain \tabularnewline
		\midrule 
		$[0.005, 0.025]$ & $85$ & $0.01114$ & $0.03604$ & $0.04345$ &\cellcolor{low} $0.30902$	\tabularnewline
		$[0.002, 0.005]$ & $81$ & $0.00338$ & $0.00421$ & $0.00747$ &\cellcolor{low} $0.80443$ \tabularnewline
		$[0.001, 0.002]$ & $25$ & $0.00142$ & $0.00035$ & $0.00174$ &\cellcolor{high} $4.09784$ \tabularnewline
		$[0.0005, 0.001]$ & $16$ & $0.00071$ & $0.00006$ & $0.00077$ &\cellcolor{high} $11.08439$
\tabularnewline	
		\bottomrule
	\end{tabular}
       }
	\caption{Experiments on equal-length-distinguishing DFA}
	\label{tbl:noisyWoSpecialProperty}
\end{table}

\begin{table}[h]
	\centering
	\scalebox{.8}{
	\begin{tabular}{l*{5}{@{\hskip1em}r}}
		\hline
		Range& \# & $d(\mathcal L(\A),\mathcal L(\A^{\leftarrow p}))$  & $d(\mathcal L(\A),\mathcal L(\A_E))$ &  $d(\mathcal L(\A^{\leftarrow p}),\mathcal L(\A_E))$ &gain \tabularnewline
		\midrule
		$[0.005, 0.025]$ & $36$ & $0.01089$ & $0.02598$ & $0.03410$ &\cellcolor{low} $0.41905$ \tabularnewline
		$[0.002, 0.005]$ & $49$ & $0.00308$ & $0.00387$ & $0.00646$ &\cellcolor{low} $0.79628$ \tabularnewline
		$[0.001, 0.002]$ & $35$ & $0.00136$ & $0.00057$ & $0.00182$ &\cellcolor{high} $2.39863$ \tabularnewline
		$[0.0005, 0.001]$ & $36$ & $0.00075$ & $0.00063$ & $0.00130$ &\cellcolor{medium} $1.18583$ \tabularnewline
		\bottomrule
	\end{tabular}
}
	\caption{Experiments on non equal-length-distinguishing DFA}
	\label{tbl:noisySpecialProperty}
\end{table}

%% file: conclusion.tex

\section{Conclusion}
\label{sec:con}

We have studied how the PAC-version of Angluin's algorithm behaves for devices which are obtained
from a DFA by introducing noise. More precisely, we have investigated whether Angluin's algorithm reduces the noise 
producing a DFA closer to the original one than the noisy device.
We have  considered three kinds of noise belonging either to random noise or to structured noise. We have shown that, on average, Angluin's algorithm
behaves well for random noise but not for structured noise. We have completed our study by establishing 
that almost surely the random noisy devices produce a non recursively enumerable language
confirming the relevance of the structural criterion for robustness of Angluin's algorithm.

There are several directions for future work.
First the algorithm could be tuned in a more precise way. In addition to  stop when the maximal number of rounds
is reached or the current automaton is declared equivalent, we could add early stopping when after some  stage with distance decreasing
the distance stabilizes. This would produce smaller DFA possibly closer to the original DFA.
At longer term,  Angluin's algorithm has no information about the original DFA. It would be interesting to introduce a priori knowledge
and design more efficient algorithms. For instance, the algorithm could take as input the maximal size of the original DFA or 
a regular language that is a superset of the original language.
In our setting the noise resulted in a noisy device which, once obtained, answers membership queries deterministically. A completely different form of noise would be that the answer to a query is randomly noisy meaning that for the same repeated query, 
different answers could occur.